\documentclass[journal,twoside,web]{ieeecolor}

\usepackage{amsthm, amsfonts}
\usepackage{algorithmic}
\usepackage{array}
\usepackage[caption=false,font=normalsize,labelfont=sf,textfont=sf]{subfig}
\usepackage{stfloats}
\usepackage{url}
\usepackage{verbatim}
\usepackage{graphicx}
\def\BibTeX{{\rm B\kern-.05em{\sc i\kern-.025em b}\kern-.08em
		T\kern-.1667em\lower.7ex\hbox{E}\kern-.125emX}}
\usepackage{balance}
\usepackage{comment}
\usepackage{enumerate}
\usepackage[utf8]{inputenc}
\usepackage{textcomp}
\usepackage{microtype}
\usepackage{xcolor}
\usepackage{cite}
\usepackage{csquotes}
\usepackage{siunitx}
\usepackage{hyperref} 
\hypersetup{colorlinks=true,
	breaklinks=true,
	urlcolor= black,
	linkcolor= blue,
	bookmarksopen=false,
	filecolor=black,
	citecolor=blue,
	linkbordercolor=blue
}
\usepackage{mathtools}
\usepackage{tcolorbox}
\newtcbox{\mybox}[1]{nobeforeafter, colframe=black, colback=white, boxrule=0.5mm, width=\linewidth, arc=0mm, boxsep=0mm, left=5mm, right=5mm}
\usepackage{multirow}
\usepackage{booktabs}
\usepackage{colortbl}
\usepackage{multirow}
\usepackage{soul} 
\usepackage{MnSymbol}%
\usepackage{wasysym}
\usepackage{lcsys}

\newtheoremstyle{mytheoremstyle} 
{\topsep}                    
{\topsep}                    
{\itshape}                   
{}                           
{\scshape}                   
{.}                          
{.5em}                       
{}  
\theoremstyle{mytheoremstyle}
\usepackage{times}
\newtheorem{theorem}{Theorem}
\newtheorem{proposition}{Proposition}
\newcommand\numeq[1]%
{\stackrel{\scriptscriptstyle(\mkern-1.5mu#1\mkern-1.5mu)}{=}}
\newtheorem{lemma}{Lemma}
\newtheorem{remark}{Remark}
\newtheorem{definition}{Definition}
\newtheorem{assumption}{Assumption}
\newtheorem{problem}{Problem}

\newtheorem*{example*}{Example}

\newcommand{\norm}[1]{\left\lVert#1\right\rVert}
\newcommand{\mat}[1]{\begin{bmatrix}#1\end{bmatrix}}

\usepackage{eso-pic} 

\begin{document}

\AddToShipoutPictureBG*{%
  \AtPageUpperLeft{%
    \setlength{\unitlength}{1mm}%
    \put(0,-12){\makebox(\paperwidth,0)[c]{\parbox{0.8\textwidth}{\centering\textcolor{gray}{\large This paper has been accepted for publication in IEEE Control Systems Letters (L-CSS) ©IEEE}}}}
  }
}

\def\BibTeX{{\rm B\kern-.05em{\sc i\kern-.025em b}\kern-.08em
    T\kern-.1667em\lower.7ex\hbox{E}\kern-.125emX}}
\markboth{\journalname, VOL. XX, NO. XX, XXXX 2017}
{Author \MakeLowercase{\textit{et al.}}: Preparation of Papers for IEEE Control Systems Letters (August 2022)}

\title{
Data-Fused MPC with Guarantees: Application to Flying Humanoid Robots
}

\author{Davide Gorbani$^{\star}$, Mohamed Elobaid$^{\star}$, Giuseppe L'Erario, \\ Hosameldin Awadalla Omer Mohamed and Daniele Pucci  
\thanks{$^{\star}$These authors contributed equally to this paper}
\thanks{Davide Gorbani, Giuseppe L'Erario, Hosameldin A. O. Mohamed and Daniele Pucci are with the Artificial and Mechanical Intelligence (AMI), Istituto Italiano di Tecnologia (IIT), Genoa, Italy (email:  {davide.gorbani@iit.it}; {giuseppe.lerario@iit.it}; {hosameldin.mohamed@iit.it};  {daniele.pucci@iit.it})}%
\thanks{Mohamed Elobaid is with the Robotics, Intelligent Systems and Control, KAUST,  Thuwal, Mecca Province, Saudi Arabia (email: {mohamed.elobaid@kaust.edu.sa}) }%
}

\maketitle
\thispagestyle{empty}

\begin{abstract}
This paper introduces a Data-Fused Model Predictive Control (DFMPC) framework that combines physics-based models with data-driven representations of unknown dynamics. Leveraging Willems’ Fundamental Lemma and an artificial equilibrium formulation, the method enables tracking of changing, potentially unreachable setpoints while explicitly handling measurement noise through slack variables and regularization. We provide guarantees of recursive feasibility and practical stability under input–output constraints for a specific class of reference signals. The approach is validated on the iRonCub flying humanoid robot, integrating analytical momentum models with data-driven turbine dynamics. Simulations show improved tracking and robustness compared to a purely model-based MPC, while maintaining real-time feasibility.\looseness=-1

Code: \url{https://github.com/ami-iit/paper_gorbani_elobaid_2025_lcss_df_mpc-ironcub}.
\end{abstract}

\begin{IEEEkeywords}
Data-driven control, Model predictive control, Robotics
\end{IEEEkeywords}

\vspace{-8pt}
\section{Introduction}
\vspace{-3pt}
\IEEEPARstart{M}{ODEL} Predictive Control (MPC) is a cornerstone of modern control theory, valued for handling multivariable systems while enforcing explicit constraints. Recent developments have extended MPC to data-driven frameworks, reducing reliance on complete parametric models. Willems' Fundamental Lemma~\cite{willem} and behavioral system theory provide a basis for representing unknown linear time-invariant (LTI) dynamics directly from measured input--output trajectories, enabling predictive control without an explicit identification step.\looseness=-1

This paradigm has led to influential methods, from early data-driven predictive controllers~\cite{yang2015data} to the widely adopted DeePC algorithm~\cite{coulson2019data}. In this line of work, non-parametric predictors are constructed from a single sufficiently rich trajectory, and stability/robustness analyses have been developed in~\cite{berberich2021tac,bongard2022robust}. However, purely data-driven approaches can be sensitive to measurement noise~\cite{breschi2023data}, where corrupted offline data may degrade performance and invalidate guarantees. They may also incur significant computational burden~\cite{zhang2023dimension}, since the optimization size grows with the data length, which can hinder real-time deployment in high-dimensional systems.\looseness=-1

\begin{figure}
    \centering
    \includegraphics[width=0.95\linewidth]{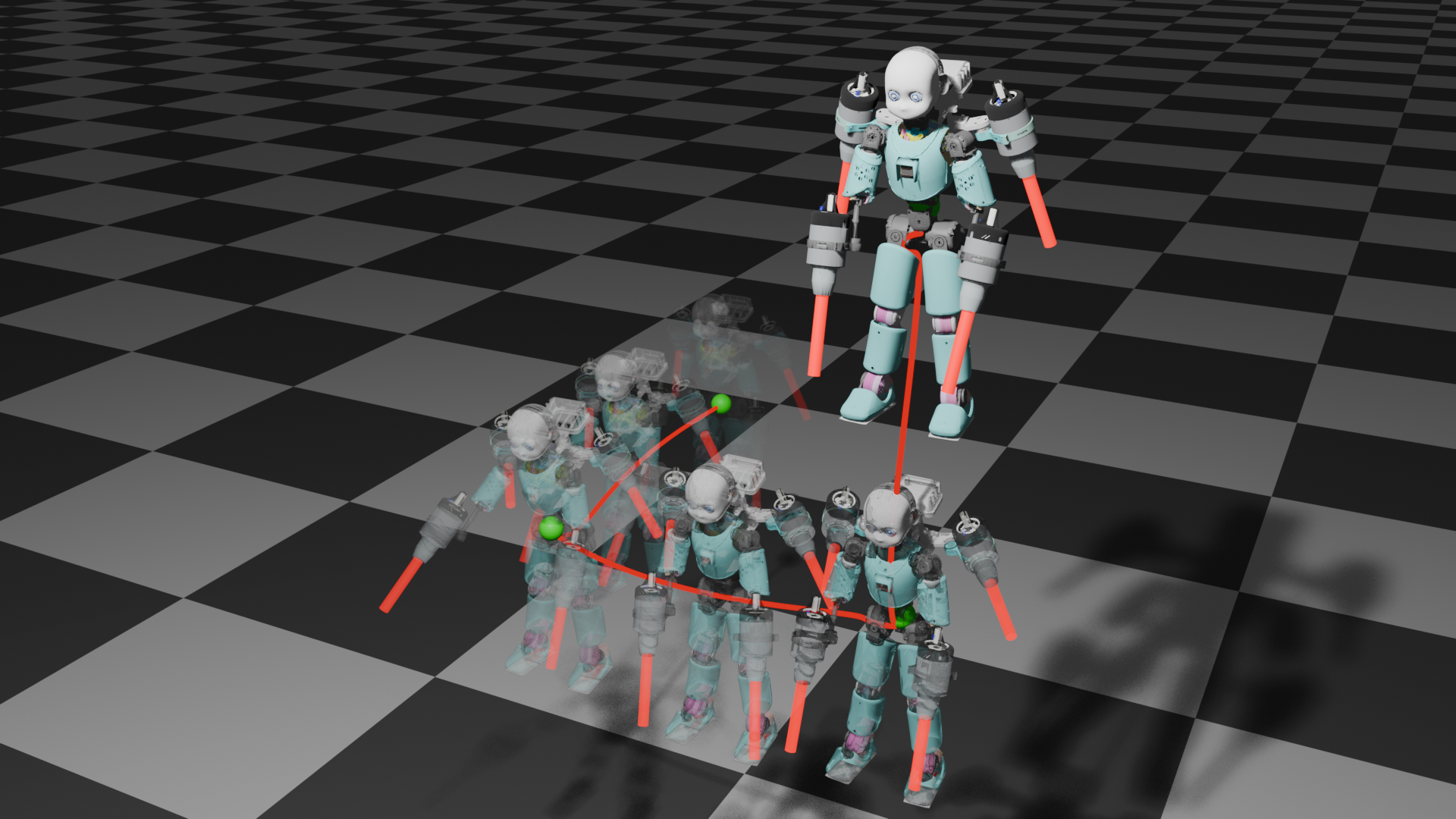}
    \caption{Snapshot of the iRonCub robot tracking a set of desired points.}
    \label{fig:snapshot_ironcub}
    \vspace{-10pt}
\end{figure}

Hybrid modeling schemes mitigate these issues by combining physics-based models for well-understood subsystems with data-driven representations for uncertain dynamics. For instance,~\cite{watson2025hybrid} incorporates model knowledge into DeePC-like formulations to reduce complexity, while~\cite{zieglmeier2025semi} proposes a semi-data-driven MPC scheme that leverages limited parametric knowledge and models residual errors. Nonetheless, the simultaneous treatment of noisy measurements, constrained operation, and reference tracking with guarantees remains challenging.\looseness=-1

In this paper, we propose a Data-Fused Model Predictive Control (DFMPC) framework that integrates model-based and data-driven dynamics, explicitly accounts for noisy measurements through slack variables and regularization, and enables tracking of changing (possibly unreachable) setpoints via artificial equilibria. The artificial equilibrium concept was introduced in~\cite{limon2008mpc} for tracking piecewise-constant references in linear MPC and was adapted to data-driven setpoint tracking in~\cite{Berberich2020Tracking}. We establish guarantees of recursive feasibility and practical stability under input--output constraints for a class of reference signals. The framework is applied to the iRonCub flying humanoid robot, combining analytical momentum models with a data-driven representation of the jet engines using an online adaptation strategy~\cite{berberich2022linear}. Through an ablation study (fixed vs.\ online Hankel updates) and benchmarks against robust semi-data-driven~\cite{zieglmeier2025semi} and model-based MPC~\cite{gorbani2025unified}, we show improved tracking while maintaining real-time computational feasibility.\looseness=-1

\vspace{-6pt}
\section{Background and Setting}
\label{sec:background}
\vspace{-3pt}

\subsection{Preliminaries and notation}
\label{sec:notation}
\vspace{-3pt}
Let $\mathbf{x}=\{x(k)\}_{k=0}^{N-1}$ denote a finite sequence of vectors with $x(k)\in\mathbb{R}^m$. We use bold lowercase for finite sequences (e.g. $\mathbf{x},\mathbf{u},\mathbf{y}$) and plain lowercase for their elements (e.g. $x(k)$, $u(k)$, $y(k)$).  For a sequence $\mathbf{x}$ and indices $a\le b$
$
\mathbf{x}_{[a,b]} \;=\; \big(x(a),\,x(a+1),\,\dots,\,x(b)\big)
$
the subsequence from time $a$ to $b$. Sequence $\mathbf{x}_1$ is appended to the tail of $\mathbf{x}_2$ by writing $\mathbf{x}_1 \oplus \mathbf{x}_2$. 
Given $\mathbf{x}=\{x(k)\}_{k=0}^{N-1}$ with $x(k)\in\mathbb{R}^m$ and an integer $L\ge1$, the order-$L$ Hankel matrix of $\mathbf{x}$ is
\[
H_L(\mathbf{x})
\;=\;
\begin{pmatrix}
x(0) & x(1) & \cdots & x(N-L)\\
x(1) & x(2) & \cdots & x(N-L+1)\\
\vdots & \vdots & \ddots & \vdots\\
x(L-1) & x(L) & \cdots & x(N-1)
\end{pmatrix},
\]
where each $x(i)$ is a column block in $\mathbb{R}^m$. Hence $H_L(\mathbf{x})\in\mathbb{R}^{mL\times (N-L+1)}$.
The sequence $\mathbf{x}$ is \emph{persistently exciting of order $L$} if and only if
$
\operatorname{rank}\big(H_L(\mathbf{x})\big)=mL.
$
Let ${\mathbf{u}=\{u(k)\}_{k=0}^{N-1}}$ and $\mathbf{y}=\{y(k)\}_{k=0}^{N-1}$ be input and output sequences of an unknown linear time-invariant (LTI) system. The pair $\{\mathbf{u},\mathbf{y}\}$ is a \emph{trajectory} of an LTI system of order $n$ if there exists a state sequence $\mathbf{x}=\{x(k)\}_{k=0}^{N-1}$ with $x(k)\in\mathbb{R}^n$ and a state $x(0)=x^\circ$ such that, for all $k=0,\dots,N-2$,
\[
x(k+1)=A x(k)+B u(k),\qquad y(k)=C x(k)+D u(k).
\]
The following instrumental result shows that a \textit{direct} nonparametric representation of an unknown LTI system can be made from a single input-output data sequence, provided that the input sequence is persistently exciting of a specific order. 

\vspace{-2pt}
\begin{theorem}[Willems' fundamental lemma]
Let $\{\mathbf{u}^d,\mathbf{y}^d\}$ be a trajectory of an LTI system of order $n$ and suppose $\mathbf{u}^d$ is persistently exciting of order $L+n$. Then any input--output sequence $\{\bar{\mathbf{u}},\bar{\mathbf{y}}\}$ of length $L$ is a trajectory of the same system if and only if there exists a vector $g\in\mathbb{R}^{N-L+1}$ such that
\begin{equation}\label{eq:fundamental-lemma}
\begin{pmatrix}
H_{L}(\mathbf{u}^d)\\[2pt]
H_{L}(\mathbf{y}^d)
\end{pmatrix} g
\;=\;
\begin{pmatrix}
\bar{\mathbf{u}}\\[2pt]
\bar{\mathbf{y}}
\end{pmatrix}, \vspace{-9pt}
\end{equation}
\hfill $\vartriangleleft$
\end{theorem}
\vspace{-7pt}

\noindent A constant pair $(u^s,y^s)$ is an \emph{equilibrium} of an LTI system (with realization $(A,B,C,D)$) if there exists $x^s\in\mathbb{R}^n$ with
\[
(I - A)x^s =  B u^s,\qquad y^s = C x^s + D u^s.
\]
Equivalently, in the data-driven setting (with $n$ the unknown system order), the constant pair $(u^s,y^s)$ repeated $n$-times form a trajectory if and only if there exists $g$ satisfying
\[
H_n(\mathbf{u}^d)\,g = \mathbf{1}_n\otimes u^s,\qquad
H_n(\mathbf{y}^d)\,g = \mathbf{1}_n\otimes y^s,
\]
where $\mathbf{1}_n$ is the $n$-vector of ones and $\otimes$ denotes the Kronecker product.
Given a (piecewise-constant) reference pair $(u_{\mathrm{ref}},y_{\mathrm{ref}})$ and positive definite weighting matrices $S\succ 0$ and $T\succ 0$, we define the \emph{optimal reachable equilibrium} as the solution of the convex problem
\begin{equation}\label{eq:reachable-equilibrium}
\begin{aligned}
J_s^\star(u_{\mathrm{ref}},y_{\mathrm{ref}}) \;=\; & \min_{u^s,y^s,g}\ \|u^s-u_{\mathrm{ref}}\|_S^2 + \|y^s-y_{\mathrm{ref}}\|_T^2\\[2pt]
\text{s.t. }\ & \begin{bmatrix}
    H_n(\mathbf{u}^d) \\
    H_n(\mathbf{y}^d)
\end{bmatrix} g = 
\begin{bmatrix}
    \mathbf{1}_n\otimes u^s \\
    \mathbf{1}_n\otimes y^s
\end{bmatrix} \\
& (u^s,y^s)\in\mathcal{U}\times\mathcal{Y},
\end{aligned}
\end{equation}
where $\|v\|_S^2 := v^\top S v$.

\vspace{-3pt}
\subsection{Modeling and Problem Statement}
\begin{figure}[!t]
    \centering
    \includegraphics[width=0.8\columnwidth]{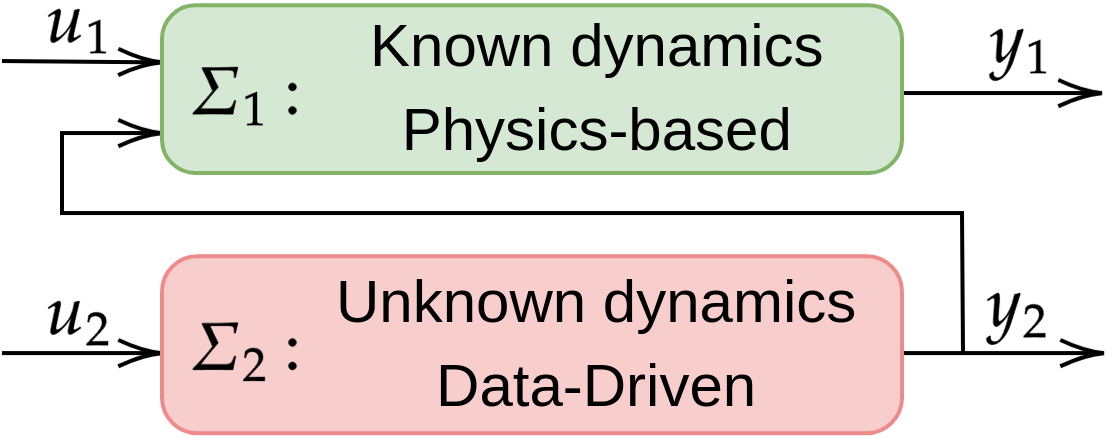}
    \vspace{-2pt}
    \caption{Structure of the considered composite system.}
    \label{fig:sysStruct}
    \vspace{-12pt}
\end{figure}
\vspace{-2pt}
We consider two interconnected subsystems with \emph{unidirectional coupling} motivated by our flying humanoid setting (the output of $\Sigma_2$ influences $\Sigma_1$, but not vice-versa as in Fig.~\ref{fig:sysStruct}):
\begin{itemize}
\item $\Sigma_{1}$ with known LTI dynamics,
\begin{align*}
    x_1(k+1) &= A_1 x_1(k) + B_1 u_1(k) + E_1 y_2(k), \notag \\
    y_1(k) &= C_1 x_1(k).
\end{align*}
\item $\Sigma_{2}$ with unknown dynamics, represented by a single input–output trajectory $\{\mathbf{u}_2^d,\mathbf{y}_2^d\}$ of length $N$.
\end{itemize}
The composite state is $z = [x_1^T, x_2^T]^T \in \mathbb{R}^{n_c}$, with $n_c = n_1 + n_2$. Inputs and outputs are \(u = [u_1^T, u_2^T]^T,\; y = [y_1^T, y_2^T]^T\). Inputs and outputs must satisfy polytopic constraints: \(u_k \in \mathcal{U},\; y_k \in \mathcal{Y}\), reflecting actuator limits and saturations.
\vspace{-3pt}
\begin{assumption}\label{ass:subsystem-properties}
Throughout the rest of this document, we assume the following hold
    \begin{enumerate}
        \item Subsystem $\Sigma_1$ is minimal, i.e., the pair $(A_1, B_1)$ is controllable and the pair $(C_1, A_1)$ is observable.
        \item The input-output behaviour of $\Sigma_2$ can be explained by a minimal realization with no feed-through, i.e., $(A_2, B_2)$ are controllable, and $(C_2, A_2)$ observable, and $D_2 = 0$.
        \item The input data trajectory $\mathbf{u}_2^d$ is persistently exciting of order $L+2\,n$, where $n = \max\{n_1, n_2\}$.
        \item System $\Sigma_2$ is affected by an additive measurement noise, i.e, $\mathbf{y}_2^d = \mathbf{y}_2^n + \mathbf{\delta}$. Moreover, Let $\epsilon = \norm{\mathbf{\delta}}_\infty$ be an upper bound on the noise. The same noise and bound apply to the measured output. $\Sigma_1$ is noise-free. \hfill $\vartriangleleft$
    \end{enumerate}
\end{assumption}

\vspace{-3pt}
\begin{remark}\label{IOSS}
Given Assumption \ref{ass:subsystem-properties}, then the composite system $\Sigma_c$ is controllable and observable. To see this, recall that $\Sigma_2$ can be explained by the triplet $(A_2,\ B_2,\ C_2)$ corresponding to an unknown minimal realization. Writing the overall system as $z(k+1) = A_c z(k) + B_c u(k)$, $y(k) = C_c z(k)$ with matrices
\[ A_c = \mat{A_1 & E_1 C_2 \\ 0 & A_2}, \quad B_c = \mat{B_1 & 0 \\ 0 & B_2}, \quad C_c = \mat{C_1 & 0 \\ 0 & C_2}, \]
the Hautus test \cite{hautus} verifies minimality of $\Sigma_c$.
A consequence of the above\cite{ioss} is that there exists a quadratic positive definite IOSS-Lyapunov function $W(z) = z^T P z$ for $\Sigma_c$ satisfying for $P > 0$ and some $c_1, \ c_2 > 0$
\begin{align*}
     W(z(k+1)) - W(z(k)) 
    &\leq -\frac{1}{2}\norm{z(k)}_2^2 \\
    &\quad + c_1\norm{u(k)}_2^2 
          + c_2\norm{y(k)}_2^2.
\end{align*}
\end{remark}
\begin{problem}\label{pr:main-problem}
    Given a piece-wise constant reference signal $\{u_{\text{ref}}, y_{\text{ref}}\}$ coming from a high-level planner, and let Assumption \ref{ass:subsystem-properties} hold, design a control input such 
    \begin{enumerate}[$(i)$]
		\item The tracking error remains bounded
		\item The closed-loop feedback system is \textit{Lyapunov stable} in a practical sense \hfill $\vartriangleleft$
	\end{enumerate} 
\end{problem}

\section{Main Results}
\label{sec:mainResults}
In this section, we detail the proposed Data-Fused Model Predictive Control scheme that solves Problem \ref{pr:main-problem}. In addition, several claims are made concerning guarantees for closed-loop performance of the proposed controller.

\subsection{The Proposed DFMPC Scheme}

At each sampling time $k$, we solve, in receding horizon fashion, the convex program
\begin{subequations}\label{eq:ocp}
\begin{equation}\label{eq:ocp-cost}
\begin{aligned}
  J_L^\star(k) &\;=\; \min_{\substack{g(k),\,\mathbf{u},\mathbf{y},\mathbf{x}_1,\\ u^s(k),\,y^s(k),\,x_1^s(k), \,\boldsymbol{\sigma}}}
  \Bigg\{ 
   \sum_{i=0}^{L-1} \Big( 
        \|y(i)-y^s(k)\|_Q^2\\
      &+ \|u(i)-u^s(k)\|_R^2 \Big)
      + \|\sigma(k)\|_\Gamma^2  
   + \|g(k)\|_\Lambda^2 \\ & 
    + \|y^s(k)-y_{\mathrm{ref}}(k)\|_T^2
    + \|u^s(k)-u_{\mathrm{ref}}(k)\|_S^2
  \Bigg\}
\end{aligned}
\end{equation}
subject to
\begin{equation}\label{eq:ocp-dynamics}
\begin{aligned}
  x_1(i+1) &= A_1 x_1(i) + B_1 u_1(i) + E_1 y_2(i),\\
  y_1(i)   &= C_1 x_1(i), \qquad i=0,\dots,L-1,\\[4pt]
  \begin{bmatrix} H_{L+n_2}(\mathbf{u}_2^d) \\[2pt] H_{L+n_2}(\mathbf{y}_2^d) \end{bmatrix} g(k)
   &= \begin{bmatrix} \mathbf{u}_2 \\[2pt] \mathbf{y}_2 + \boldsymbol{\sigma} \end{bmatrix},
\end{aligned}
\end{equation}
with initialization
\begin{equation}\label{eq:ocp-init}
\begin{aligned}
  x_1(0) &= x_1^\circ(k),\\
(u_2, \ y_2)_{[-n_2,-1]} &= (u_2^\circ(k), \ y_2^\circ(k)),
\end{aligned}
\end{equation}
terminal conditions
\begin{equation}\label{eq:ocp-terminal}
\begin{aligned}
  x_1(L) &= x_1^s(k),\\
  (\mathbf{I}_{n_1} - A_1)x_1^s(k) &= B_1u_1^s(k) + E_1y_2^s(k), \\
  (u_2,\ y_2)_{[L:L+n_2-1]} &=
       (\mathbf{1}_{n_2}\otimes u_2^s(k),\ \;\mathbf{1}_{n_2}\otimes y_2^s(k)),
\end{aligned}
\end{equation}
and hard input–output constraints
\begin{equation}\label{eq:ocp-constraints}
    u(i), {u^s}\in\mathcal{U},\quad y(i),{y^s}\in\mathcal{Y},\qquad i=0,\dots,L-1.
\end{equation}
\end{subequations}

Where we drop the time index $k$ for compactness e.g, writing $y(i) := y(k+i|k)$ and similarly for other sequence variables. Here $Q,R,\Gamma,\Lambda\succ0$ are weighting matrices, $\boldsymbol{\sigma}$ is a slack variable compensating measurement noise in $\Sigma_2$, and $S,T\succ0$ are as in \eqref{eq:reachable-equilibrium}. 
The cost \eqref{eq:ocp-cost} penalizes deviation from an artificial equilibrium $(u^s,y^s)$ close to the reference, penalizes slacks; the regularization of the term $g$ enhances numerical conditioning and mitigates overfitting to measurement noise present in the Hankel matrix~\cite{berberich2021tac}, improving the robustness of the prediction. 
Constraints \eqref{eq:ocp-dynamics} describe the composite prediction model: $\Sigma_1$ via its known matrices $(A_1,B_1,C_1,E_1)$, and $\Sigma_2$ via the Hankel-based direct representation from data. 
\eqref{eq:ocp-init}--\eqref{eq:ocp-terminal} enforce past consistency and a terminal equilibrium tail, respectively, and \eqref{eq:ocp-constraints} enforces polytopic input–output constraints. 

Compared to \cite{Berberich2020Tracking}, two main differences arise: (i) we treat a ``hybrid'' setting, where only part of the system is modeled parametrically, while the other part is described directly by data; (ii) we explicitly account for measurement noise both in the offline data $\mathbf{y}_2^d$ and in the online measurements, cf. Assumption~\ref{ass:subsystem-properties}. 
We adopt the following standard assumption.

\begin{assumption}\label{ass:prediction-horizon-length}
The OCP \eqref{eq:ocp} is feasible at the initial time. Moreover, the horizon length satisfies $L\ge2n$, where $n = \max\{n_1, n_2\}$. \hfill $\vartriangleleft$
\end{assumption}

At this point, we are in a position to state the following intermediate and helpful claim
\begin{proposition}[Output prediction error]\label{prop:prediction-error}
    Denote by $\mathbf{u}^\star(k)$ the optimal input sequence and by $\mathbf{y}^\star(k)$ the predicted optimal output sequence from solving (\ref{eq:ocp}) at time step $k$. Let $y(k+i)$ be the actual output vector of the composite system at time $k+i$ and let $y^\star_i(k)$ be the corresponding predicted optimal output, computed at time $k$. Then for $i \in \{0, \dots, L-1\}$
    \begin{align}\label{eq:error-bound}
        \| y(k+i) - y_i^\star(k) \|_\infty & \leq\Big(1 + \tilde{c}_{\Sigma_1}\Big) \Big(\tilde{c}_{\Sigma_2}  \big( \epsilon [\|g^\star(k)\|_1 + 1] \notag\\
        + &\|{\boldsymbol{\sigma}^\circ}^\star(k) \|_\infty \big) + \epsilon\|g^\star(k)\|_1 + \|\boldsymbol{\sigma}^\star(k)\|_\infty\Big)
    \end{align}
   where, ${\boldsymbol{\sigma}^\circ}^\star(k)$ are the initial $n_2$ values of $\boldsymbol{\sigma}^\star(k)$, $\mathcal{O}_{n_2}^\#$ is the Pseudo-Inverse of the observability matrix of $\Sigma_2$, and
   \begin{align*}
       \tilde{c}_{\Sigma_1} &= \max_{i \in \{0,..,L-1\}} \left( \|C_1\|_\infty \sum_{j=0}^{i-1} \|A_1^{i-1-j}\|_\infty \|E_1\|_\infty \right), \\
       \tilde{c}_{\Sigma_2} &= \max_{i \in \{0,..,L-1\}} \left( \|C_2\|_\infty \|A_2^{i}\|_\infty \norm{\mathcal{O}_{n_2}^\#}_\infty \right).
   \end{align*}
\end{proposition}
\begin{proof}
Denote at time step $k+i$
\begin{align*}
    e_{y_1}(k+i) &:= y_1(k+i) - y_{1,i}^\star(k) \\
    e_{y_2}(k+i) &:= y_2(k+i) - y_{2,i}^\star(k),
\end{align*}
Following the logic of Lemma 2 in \cite{berberich2021tac}, for $\Sigma_2$ we have
\begin{equation}
    H_{L+n_2}(\mathbf{y}_2^d) g^\star(k) = \mathbf{y}_2^\star(k) + \boldsymbol{\sigma}^\star(k) ,\label{eq:proof_s2_dyn}
\end{equation}
re-arranging
\begin{equation*}
    H_{L+n_2}(\mathbf{y}_2^n) g^\star(k) + H_{L+n_2}(\mathbf{\delta}) g^\star(k) = \mathbf{y}_2^\star(k) + \boldsymbol{\sigma}^\star(k)
\end{equation*}
The noise-free trajectory $H_{L+n_2}(\mathbf{y}_2^n) g^\star(k)$ is a valid trajectory of $\Sigma_2$. Since this trajectory is generated by the same input sequence $\mathbf{u}_2^\star(k)$ as the actual trajectory, their difference is a zero-input response of $\Sigma_2$ to an initial state error. Rearranging the terms, the prediction error for $\Sigma_2$ is
\begin{align*}
    e_{y_2}(k+i) &= y_2(k+i) - y_{2,i}^\star(k) \\
                 &= y_2(k+i) - \big( (H_{L+n_2}(\mathbf{y}_2^n) g^\star(k))_i \\ &+ (H_{L+n_2}(\mathbf{\delta}) g^\star(k))_i - \sigma_i^\star(k) \big) \\
                 &:= \bar{y}_{2,i}(k) - (H_{L+n_2}(\mathbf{\delta}) g^\star(k))_i + \sigma_i^\star(k).
\end{align*}
Taking the norm and using the triangle inequality
\begin{align}
    \|e_{y_2}(k+i)\|_\infty &\leq \|\bar{y}_{2,i}(k)\|_\infty + \|(H_{L+n_2}(\mathbf{\delta}) g^\star(k))_i\|_\infty \notag \\&+ \|\sigma_i^\star(k)\|_\infty \label{eq:proof_s2_triangle}.
\end{align}
The term $\|\bar{y}_{2,i}(k)\|_\infty$ is the system's response to an initial error. This initial error is caused by the mismatch between the initial conditions of the system and those implied by $H_{L+n_2}(\mathbf{y}_2^n) g^\star(k)$, which in turn arises from the problem constraints, the noise $\mathbf{\delta}$, and the initial part of the slack variable, denoted $\sigma^\circ(k)$. This response is bounded by the gain of the system $\Sigma_2$, denoted $\tilde{c}_{\Sigma_2}$, multiplied by the magnitude of the initial error. The initial error sources include online measurement noise (bounded by $\epsilon$), historical noise ($\epsilon\|g^\star(k)\|_1$), and the initial slack ($\|{\sigma^\circ}^\star(k)\|_\infty$). The second term in \eqref{eq:proof_s2_triangle} is bounded element-wise by $\epsilon \|g^\star(k)\|_1$. Combining these facts gives \cite{berberich2021tac}
\begin{align}
    \|e_{y_2}(k+i)\|_\infty &\leq \tilde{c}_{\Sigma_2} \big( \epsilon [\|g^\star(k)\|_1 + 1] + \|{\sigma^\circ}^\star(k) \|_\infty \big) \notag \\&+ \epsilon\|g^\star(k)\|_1 + \|\mathbf{\sigma}^\star(k)\|_\infty \label{eq:proof_s2_bound}.
\end{align}
Now note that for $\Sigma_1$ we have
\begin{align*}
    x_1(k+i+1) &= A_1 x_1(k+i) + B_1 u_{1,i}^\star(k) + E_1 y_2(k+i) \\
    x_{1,i+1}^\star(k) &= A_1 x_{1,i}^\star(k) + B_1 u_{1,i}^\star(k) + E_1 y_{2,i}^\star(k).
\end{align*}
Denote $e_{x_1}(k+i) := x_1(k+i) - x_{1,i}^\star(k)$ and
\begin{equation*}
    e_{x_1}(k+i+1) = A_1 e_{x_1}(k+i) + E_1 e_{y_2}(k+i),
\end{equation*}
The initial condition constraint of the problem implies $x_1(k) = x_{1,0}^\star(k)$ (given no noise on $\Sigma_1$), so $e_{x_1}(k)=0$. The solution to this system is the zero-state response to the input $e_{y_2}(k+j)$, given by the convolution sum
\begin{equation*}
    e_{x_1}(k+i) = \sum_{j=0}^{i-1} A_1^{i-1-j} E_1 e_{y_2}(k+j).
\end{equation*}
The output error for $\Sigma_1$ is $e_{y_1}(k+i) = C_1 e_{x_1}(k+i)$. Taking the norm
\begin{align*}
    \|e_{y_1}(k+i)\|_\infty &\leq \|C_1\|_\infty \bigg\| \sum_{j=0}^{i-1} A_1^{i-1-j} E_1 e_{y_2}(k+j) \bigg\|_\infty \\
    &\leq \tilde{c}_{\Sigma_1} \left( \max_{j \in \{0,..,L-1\}} \|e_{y_2}(k+j)\|_\infty \right).
\end{align*}
The total system prediction error is $\|y(k+i) - y_i^\star(k)\|_\infty = \max(\|e_{y_1}(k+i)\|_\infty, \|e_{y_2}(k+i)\|_\infty)$. This is bounded by the sum of the individual bounds, which is a valid upper bound
\begin{align*}
    \|y(k+i) - y_i^\star(k)\|_\infty &\leq \|e_{y_1}(k+i)\|_\infty + \|e_{y_2}(k+i)\|_\infty \\
    &\leq \tilde{c}_{\Sigma_1} \left( \max_{j} \|e_{y_2}(k+j)\|_\infty \right) \\&+ \left( \max_{j} \|e_{y_2}(k+j)\|_\infty \right) \\
    &\leq (1 + \tilde{c}_{\Sigma_1}) \left( \max_{j} \|e_{y_2}(k+j)\|_\infty \right).
\end{align*}
Substituting \eqref{eq:proof_s2_bound} completes the proof. 
\end{proof}

From Proposition \ref{prop:prediction-error}, it is clear that to establish recursive feasibility, one should modify the output constraints (\ref{eq:ocp-constraints}), invoking some constraints tightening mechanism given the bound (\ref{eq:error-bound}). Alternatively, one could rely on some inherent property of the MPC formulation, given some stricter requirements on the reference signal and the regularization penalties on the cost. We will attempt the latter approach in the following subsection.  
\subsection{Recursive Feasibility}
Before delving into the technical details, the following definition is needed.

\begin{definition}[Signed-distance to boundary]\label{def:distance-to-set}
Let the polytopic output constraint set $\mathcal{Y} \subset \mathbb{R}^{p}$ be defined by $m_y$ linear inequalities, i.e.,
\[
\mathcal{Y} = \{y \in \mathbb{R}^{p} \mid E_y y \leq e_y \},
\]
where $E_y \in \mathbb{R}^{m_y \times p}$ and $e_y \in \mathbb{R}^{m_y}$.  
The $j$-th row of $E_y$ is denoted by $E_{y,j}$.  
The boundary of the set is denoted $\partial\mathcal{Y}$.  
The signed Euclidean distance of a point $y$ to the boundary is
\[
\mathrm{dist}(y, \partial\mathcal{Y}) \;=\; \min_{j \in \{1,\dots,m_y\}} \frac{e_{y,j} - E_{y,j}y}{\|E_{y,j}\|_2}.
\] \hfill \hfill $\vartriangleleft$
\end{definition}
Definition \ref{def:distance-to-set} establishes a signed distance for any point w.r.t. the boundary of the output constraint sets. Namely, the distance being positive implies the point is in the interior of the set, zero implies lying on the boundary, and negative implies being outside the constraints set (constraints violation).
\begin{lemma}[Violating output distance]\label{lem:geom-noV}
Let $\mathcal{Y} = \{y \in \mathbb{R}^p : E_y y \le e_y\}$ be a polytope as in (\ref{eq:ocp-constraints}).  
Fix $d>0$ and $\kappa \ge 0$ with $\sqrt{p}\,\kappa<d$.  
Suppose a predicted output $y^\star$ satisfies
\[
\mathrm{dist}(y^\star,\partial\mathcal{Y}) \;\ge\; d,
\]
and an actual output $y$ satisfies
\[
\|y - y^\star\|_\infty \;\le\; \kappa.
\]
Then:
\begin{enumerate}[$(i)$]
\item $\mathrm{dist}(y,\partial\mathcal{Y}) \;\ge\; d - \sqrt{p}\,\kappa$.
\item For any $\tilde y \notin \mathcal{Y}$, one has;
$\|\tilde y - y\|_2 \;>\; d - \sqrt{p}\,\kappa.
$
\end{enumerate}
\end{lemma}

\begin{proof}
To show $(i)$, for each facet $j$, define $\phi_j(y) := (e_{y,j} - E_{y,j}y)/\|E_{y,j}\|_2$.  
Then
\[
\phi_j(y) = \phi_j(y^\star) + \frac{E_{y,j}(y^\star-y)}{\|E_{y,j}\|_2}.
\]
Since $\|y-y^\star\|_\infty \le \kappa$, we have 
\[
|E_{y,j}(y^\star-y)| \;\le\; \|E_{y,j}\|_2 \|y^\star-y\|_2 
\;\le\; \|E_{y,j}\|_2 \sqrt{p}\,\kappa,
\]
where we use $\ \norm{v}_2 \le \sqrt{p}\norm{v}_\infty$ for any $v \in \mathbb{R}^p$. Thus $\phi_j(y) \ge \phi_j(y^\star) - \sqrt{p}\,\kappa$.  
Taking the minimum over $j$ yields 
\(
\mathrm{dist}(y,\partial\mathcal{Y}) \ge d - \sqrt{p}\,\kappa.
\)
Concerning $(ii)$, let $\tilde y \notin \mathcal{Y}$. Then there exists a facet $j^\ast$ with $e_{y,j^\ast} - E_{y,j^\ast}\tilde y < 0$.  
From $(i)$, $e_{y,j^\ast} - E_{y,j^\ast}y \ge (d-\sqrt{p}\,\kappa)\|E_{y,j^\ast}\|_2$.  
Subtracting gives
\[
E_{y,j^\ast}(\tilde y - y) \;>\; (d-\sqrt{p}\,\kappa)\,\|E_{y,j^\ast}\|_2.
\]
By Cauchy–Schwarz, 
\(
\|\tilde y-y\|_2 > d - \sqrt{p}\,\kappa.
\) 
\end{proof}

\begin{proposition}[Feasible candidate under a safe reference]
\label{prop:feasible-candidate}
Suppose Assumptions \ref{ass:subsystem-properties}--\ref{ass:prediction-horizon-length} hold and Problem \ref{pr:main-problem} is feasible at time $k$ for $z(k)\in\mathcal{Z}_f$, a compact feasible set. Let $J_L^\star(k)\le V_{\max}$\footnote{$V_{\max}$ will be better characterized in the following subsection} for all $z(k)\in\mathcal{Z}_f$, and let $d_{\mathrm{ref}}=\mathrm{dist}(y_{\mathrm{ref}},\partial\mathcal{Y})$. Define, whenever possible
\[
d_{\mathrm{safe}}:=d_{\mathrm{ref}}-\Big(\sqrt{\tfrac{V_{\max}}{\lambda_{\min}(Q)}} + \sqrt{\tfrac{V_{\max}}{\lambda_{\min}(T)}}\Big)>0 .
\]
Let $n = \max\{n_1, n_2\}$, and assume $\lambda_{\min}(\Lambda),\lambda_{\min}(\Gamma)>0$.  
Then there exists $\tilde c_e \geq 0$ such that, whenever
\[
d_{\mathrm{safe}}>\sqrt{p}\,\tilde c_e,
\]
then:
\begin{enumerate}
\item[\emph{(a)}] The \emph{executed} outputs over the next $n$ steps satisfy the hard constraints (\ref{eq:ocp-constraints}), i.e.
\[
\mathrm{dist}\big(y(k+i),\partial\mathcal{Y}\big)\  >\ 0,
\quad i=0,\ldots,n-1.
\]
\item[\emph{(b)}] At time $k{+}n$ there exists a candidate optimizer that satisfies all constraints of Problem~\ref{pr:main-problem}.
\end{enumerate}
\end{proposition}
\begin{proof}
Let $(g^\star,\mathbf u^\star,\mathbf y^\star,\mathbf x_1^\star,\boldsymbol\sigma^\star,u^{s\star},y^{s\star})$ be the optimizer at time $k$, with $J_L^\star(k)\le V_{\max}$. 
For every $i=0,\ldots,L-1$, positive definiteness of $Q,T$ and $J_L^\star(k)\le V_{\max}$ imply
\begin{align*}
\|y^\star_i(k)-y_{\mathrm{ref}}\|_2
&\le \|y^\star_i(k)-y^{s\star}(k)\|_2 + \|y^{s\star}(k)-y_{\mathrm{ref}}\|_2 \\
&\le \sqrt{\tfrac{V_{\max}}{\lambda_{\min}(Q)}} + \sqrt{\tfrac{V_{\max}}{\lambda_{\min}(T)}}.
\end{align*}
Since $\mathrm{dist}(\cdot,\partial\mathcal Y)$ is 1-Lipschitz and $d_{\mathrm{ref}}=\mathrm{dist}(y_{\mathrm{ref}},\partial\mathcal Y)$,
\[
\mathrm{dist}\!\left(y^\star_{i}(k),\partial\mathcal{Y}\right)\ \ge\ d_{\mathrm{safe}}.
\]
By Proposition~\ref{prop:prediction-error};
\begin{align}\label{eq:C1_prop}
\|y(k+i)-y^\star_{i}(k)\|_\infty
&\le (1+\tilde c_{\Sigma_1})\Big[\tilde c_{\Sigma_2}\big(\epsilon[\|g^\star\|_1 + 1 ] +\|{\sigma^\star}^\circ\|_\infty\big) \notag \\& 
+ \epsilon\|g^\star\|_1+\|\boldsymbol\sigma^\star\|_\infty\Big] \notag \\
&\leq (1+\tilde c_{\Sigma_1})\Big[\tilde c_{\Sigma_2}\Big(\epsilon\big(\sqrt{\tfrac{N V_{\max}}{\lambda_{\min}(\Lambda)}}+1\big)
\notag\\&+ \tfrac{\sqrt{V_{\max}}}{\sqrt{\lambda_{\min}(\Gamma)}}\Big) 
+ \epsilon\sqrt{\tfrac{N V_{\max}}{\lambda_{\min}(\Lambda)}} + \tfrac{\sqrt{V_{\max}}}{\sqrt{\lambda_{\min}(\Gamma)}}\Big] \notag \\
&:= \tilde c_e ,
\end{align}
where we used
\begin{align*}
\|g^\star\|_1&\le \sqrt{\tfrac{N V_{\max}}{\lambda_{\min}(\Lambda)}},\quad
\|\boldsymbol\sigma^\star\|_\infty\le \sqrt{\tfrac{V_{\max}}{\lambda_{\min}(\Gamma)}}.
\end{align*}
Applying Lemma~\ref{lem:geom-noV} with $d=d_{\mathrm{safe}}$ and $\kappa=\tilde c_e$ proves $(a)$.

\medskip\noindent
To show $(b)$, we now construct a candidate at time $k{+}n$, and show that it is feasible. 
Set the candidate equilibrium $\widehat u^s:=u^{s\star}(k)$, $\widehat y^s:=y^{s\star}(k)$. 

\smallskip
\noindent \emph{Subsystem $\Sigma_2$.} For time $[-n, L-2n-1]$ we follow a shift-and-append strategy.  Denote by $\bar{y}_2$ the trajectory resulting from the open-loop application of $\mathbf{u}^\star_2$ with consistent initialization $(u_2^\circ, \ y_2^\circ)_{[-n,-1]}$. For the first $L-2n$ steps, we set a candidate output $\widehat{y}_2 = \bar{y}_2$. Note that at the tail $n$ steps, and by Proposition \ref{prop:prediction-error} we have $\norm{\bar y(k+i) - y_{2,i}^\star(k)}_\infty \le \tilde c_e$, and introduce the $\Sigma_2$ internal state $\bar{x}_{2}$ consistent with $(\mathbf{u}_2^\star, \bar{y})$ in some minimal realization, and let $x_{2}^{\text{sr}}$ be the equilibrium state corresponding to $( {u^s_2}^\star, \ {y^s_2}^\star)$. For small noise, and since at the tail $\mathbf{y}_2^\star = y_2^s$ we have at time $L-n+1$ that $\norm{\bar{x}_2(k) - x_2^{\text{sr}}}_2 \le r$ for a small $r \geq 0$. By minimality (see Remark \ref{IOSS}), there exists an input-output trajectory $(\tilde u_2, \tilde y_2)$ such that the corresponding state $\bar{x}_2$ approaches its equilibrium $(x_2^{\text{sr}})$ and satisfying
\[
 \norm{\begin{pmatrix}\tilde u_{2, [0, L-1]} \\ \tilde y_{2, [0, L-1]} \end{pmatrix}}_2^2 \leq \tilde c_{x_2}\norm{\bar{x}_2(k) - x_2^{\text{sr}}}_2^2,
\]
for some $\tilde c_{x_2}> 0$. Form the candidate input-output $(\widehat{u}_2, \ \widehat{y}_2)$
which is a valid trajectory for $\Sigma_2$. Denote its corresponding internal state in some minimal realization $\widehat{x}_{2}(k)$ and set  
\[
\widehat g := (H_{ux}^d)^\# \begin{pmatrix} u_{2, [-n, 1]}^\circ \oplus \widehat{\mathbf u}_2 \\\widehat{x}_{2, [-n, 1]}\end{pmatrix},
\]
where $H_{ux}^d$ collects input Hankel blocks and corresponding state (see \cite[eq~7]{berberich2021tac}). Then the candidate slack satisfy;
\[
\widehat{\mathbf y}_2 = H_{L+n_2}(y_2^d)\,\widehat g - \widehat{\boldsymbol\sigma},
\]
\smallskip
\noindent \emph{Subsystem $\Sigma_1$.} Define the shifted–appended input
\[
\widehat{\mathbf u}_1 := \mathbf u^\star_{1,[n:L-1]}(k)\ \oplus\ (\mathbf 1_{n}\otimes \widehat u^s_1).
\]
Initialize with the actual state $x_1(k{+}n_2)$ and roll out the known model to obtain $(\widehat{\mathbf x}_1,\widehat{\mathbf y}_1)$.

\noindent
By construction, the candidate satisfies Hankel equalities, initialization, terminal equalities, and input bounds for $\Sigma_c$ . What remains is output constraints. To that end, we establish some useful bounds.

\medskip\noindent
\emph{Bound on $\widehat g$.} Since $H_{ux}^d$ has full row rank, with $c_{\mathrm{pe}}:=\|(H_{ux}^d)^\#\|_2^2$ we obtain
\begin{align*}
\|\widehat g\|_2^2 \;&\le\; c_{\mathrm{pe}}\,\Big( \norm{\widehat{\mathbf{u}_2}}_2 + \norm{\widehat x_2}_2\Big) \\
&\leq  c_{\mathrm{pe}}\,\Big( \tilde c_{x_2} \norm{\bar x_2 - x_2^s}_2^2 + \norm{\tilde{A}_2}_2^2 \norm{\xi}_2^2 \Big) ,
\end{align*}
where $\tilde A_2$ is s.t.
\begin{align*}
    \begin{pmatrix}
        u_{2, [-n, 1]}^\circ  \\\widehat{x}_{2, [-n, 1]}
    \end{pmatrix} &= \underbrace{\begin{pmatrix}
        I & 0 \\ \star & \mathcal{O}^\#
    \end{pmatrix}}_{\tilde A_2} \underbrace{\begin{pmatrix}
         u_{2, [-n, 1]} \\  y_{2, [-n, 1]}
    \end{pmatrix}}_\xi,
\end{align*}

with $\phi x := \xi$ for some linear transformation $\phi$.

\smallskip
\noindent \emph{Bound on $\widehat\sigma$.} Write $H_{L+n_2}(y_2^d)=H_y^n+H_\delta$, where $H_\delta=H_{L+n_2}(\varepsilon^d)$. Then by definition of the slack;
\[
\|\widehat{\boldsymbol\sigma}_2\|_\infty \le c_\delta\|\widehat g\|_2 + \sqrt{n_2}\,\epsilon,
\]
with $c_\delta:=\|H_\delta\|_2$.

\medskip\noindent
Finally, compare the candidate output with the shifted optimal prediction $\tilde{\mathbf y}^\star:=\mathbf y^\star_{[n_2:L-1]}(k)\oplus(\mathbf 1_{n_2}\otimes y^{s\star})$. By applying Proposition~\ref{prop:prediction-error} to both sequences and using the triangle inequality,
\[
\max_{0\le i<L}\ \|\widehat y(k{+}n_2{+}i)-\tilde y^\star_i\|_\infty \le \tilde c_e^{(\star)}+\tilde c_e^{(\hat\cdot)}=:\tilde c_e^{\mathrm{new}},
\]
where
\[
\tilde c_e^{(\hat\cdot)}=(1+\tilde c_{\Sigma_1})\Big[\tilde c_{\Sigma_2}\big(\epsilon(\|\widehat g\|_1+1)+\|\widehat\sigma^\circ\|_\infty\big) + \epsilon\|\widehat g\|_1+\|\widehat{\boldsymbol\sigma}_2\|_\infty\Big].
\]
Since $\tilde y^\star$ satisfies $\mathrm{dist}(\tilde y^\star_i,\partial\mathcal Y)\ge d_{\mathrm{safe}}$, Lemma~\ref{lem:geom-noV} with $\kappa=\tilde c_e^{\mathrm{new}}$ gives
\[
\mathrm{dist}\!\left(\widehat y(k{+}n_2{+}i),\partial\mathcal{Y}\right)\ \ge\ d_{\mathrm{safe}}-\sqrt{p}\,\tilde c_e^{\mathrm{new}}>0,\quad i=0,\ldots,L-1.
\]
Thus the candidate output constraints at time $k{+}n_2$ are also satisfied. This proves $(b)$. 
\end{proof}

Proposition~\ref{prop:feasible-candidate} is very restrictive in the sense that it limits the class of references for which one can guarantee the existence of a feasible candidate to those that are well-inside the output constraints set by a \textit{safety margin}. The benefit of not employing constraint tightening is that the MPC remains a convex quadratic program that can be solved efficiently.


Proposition \ref{prop:feasible-candidate} alone is not enough to establish recursive feasibility. Instead, it only establishes the existence of a feasible candidate for which the output constraints are also respected by the actual output of the plant after $n$ steps. It does not ensure that no other candidate, for which the optimal predicted output is feasible but the actual measured one is not, is not chosen.
The following claim addresses this point.

\begin{proposition}[No cheaper un-safe alternative]
\label{prop:recursive-feasibility}
In addition to the hypotheses of Proposition~\ref{prop:feasible-candidate}, suppose the cost-dominance inequality
\begin{align} 
\label{eq:dominance_prop}
\lambda_{\min}(Q)&\big(d_{\mathrm{safe}}-\sqrt{p}\,\tilde c_e\big)^2
\;\ge\;
\gamma_H^2(\beta+\bar \xi)^2 \;+\; V_{\max}
\; \notag \\&+\; \lambda_{\max}(\Gamma)\Big( A_\beta(\beta+\bar \xi) + A_V\sqrt{V_{\max}} \Big)^2,
\end{align}
holds, where
\begin{align*}
\beta &:= \sqrt{\frac{V_{\text{max}}}{\lambda_{\min}(R)}}, \ \gamma_H:=\sqrt{\lambda_{\max}(\Lambda)c_{\mathrm{pe}}}, \ \alpha_H := \norm{H_L(\mathbf{y}_2^d)}_2\\
 A_\beta &:= \frac{\gamma_H}{\sqrt{\lambda_{\min}(\Lambda)}}\Big( c_\delta\epsilon + \sqrt{L}\,\tilde c_{\Sigma_2}\,\alpha_H\sqrt{N}\Big),
\\
A_V &:= \sqrt{L}\,\tilde c_{\Sigma_2}\Big( \frac{\alpha_H\sqrt{N}}{\sqrt{\lambda_{\min}(\Lambda)}} + \frac{\epsilon\sqrt{N}}{\sqrt{\lambda_{\min}(\Lambda)}} \Big)
+ \frac{\sqrt{L}\,\tilde c_{\Sigma_2}}{\sqrt{\lambda_{\min}(\Gamma)}} .
\end{align*}
Then the optimizer at time $k{+}n_2$ selects a trajectory whose predicted outputs lie in $\mathcal Y$, and Problem~\ref{pr:main-problem} is $n_2$-step recursively feasible for all $z(0)\in\mathcal{Z}_f$.
\end{proposition}

\begin{proof}
Let $\widehat z$ be the feasible candidate from Proposition~\ref{prop:feasible-candidate}. Define the safety margin
\[
\Delta \ :=\ d_{\mathrm{safe}}-\sqrt{p}\,\tilde c_e\ > 0.
\]
By Lemma~\ref{lem:geom-noV} with $d=d_{\mathrm{safe}}$ and $\kappa=C_1$, any violating output $\tilde y\notin\mathcal Y$ must satisfy
\[
\|\tilde y-\widehat y(k{+}n_2{+}i)\|_2\ \ge\ \Delta,
\]
hence, for $T$ positive definite, 
\[
\ell_y(\tilde y)-\ell_y\!\big(\widehat y(k{+}n_2{+}i)\big)\ \ge\ \lambda_{\min}(Q)\,\Delta^2.
\]
Summing over the horizon yields a net \emph{output-cost} penalty $\ge \lambda_{\min}(Q)\Delta^2$ for any unsafe feasible point relative to the candidate.
From the construction in Proposition~\ref{prop:feasible-candidate}
\[
\ell_u(\widehat u)\le V_{\max},\qquad
\ell_g(\widehat g)=\|\widehat g\|_\Lambda^2 \le \gamma_H^2(\beta+\bar \xi)^2.
\]
Moreover, writing $\widehat\sigma=H_\delta\widehat g+(H_y^n\widehat g-\widehat y_2)$ and using
$\|H_\delta\widehat g\|_2\le c_\delta\epsilon\,\|\widehat g\|_2
\le \frac{c_\delta\epsilon\,\gamma_H(\beta+\bar\xi)}{\sqrt{\lambda_{\min}(\Lambda)}}$ together with
\begin{align*}
\|H_y^n\widehat g-\widehat y_2\|_2 \ &\le\ \sqrt{L}\,\tilde c_{\Sigma_2}\,\Big(\alpha_H\sqrt{N}\frac{\|\widehat g-g^\star\|_\Lambda}{\sqrt{\lambda_{\min}(\Lambda)}} \\&+ \epsilon\sqrt{N}\frac{\|g^\star\|_\Lambda}{\sqrt{\lambda_{\min}(\Lambda)}} + \frac{\|\boldsymbol\sigma^\star\|_2}{\sqrt{\lambda_{\min}(\Gamma)}}\Big),
\end{align*}
and the bounds $\|\widehat g\|_\Lambda \le \gamma_H(\beta+\bar\xi)$, $\|g^\star\|_\Lambda\le \sqrt{V_{\max}}$, $\|\boldsymbol\sigma^\star\|_2\le \sqrt{V_{\max}/\lambda_{\min}(\Gamma)}$, we obtain
\begin{align*}
\|\widehat\sigma\|_2 \ &\le\ A_\beta(\beta+\bar \xi)+A_V\sqrt{V_{\max}} \Rightarrow
\\ 
\ell_\sigma(\widehat\sigma)&=\|\widehat\sigma\|_\Gamma^2\ \le\ \lambda_{\max}(\Gamma)\Big( A_\beta(\beta+\bar \xi) + A_V\sqrt{V_{\max}} \Big)^2.
\end{align*}

\noindent For any unsafe feasible trajectory $\tilde z$ at time $k{+}n_2$,
\begin{align*}
J(\tilde z)-J(\widehat z)
&\ge \underbrace{\lambda_{\min}(Q)\Delta^2}_{\text{output penalty}}
\ -\ \underbrace{\gamma_H^2(\beta+\bar \xi)^2}_{\ell_g(\widehat g)}\ -\ \underbrace{V_{\max}}_{\ell_u(\widehat u)} \\ &
\ -\ \underbrace{\lambda_{\max}(\Gamma)\Big( A_\beta(\beta+\bar \xi) + A_V\sqrt{V_{\max}} \Big)^2}_{\ell_\sigma(\widehat\sigma)}.
\end{align*}
If \eqref{eq:dominance_prop} holds, the right-hand side is non-negative, hence the optimizer at time $k{+}n_2$ selects a safe trajectory. 
\end{proof}

\subsection{Practical Exponential Stability}
\label{subsec:stability}
We establish a Lyapunov candidate relying on the optimal cost function and an IOSS property of the composite system. First, we establish some quadratic bounds on the optimal cost.

\begin{lemma}[optimal cost properties]\label{lem:optimal-cost}
Let Assumptions \ref{ass:subsystem-properties}--\ref{ass:prediction-horizon-length} hold, and define the offset value
$
\widetilde J^\star(\xi)\;:=\;J_L^\star(\xi)\;-\;J_s^\star(u_{\rm ref},y_{\rm ref}),
$
where $J_L^\star$ is the optimal value of \eqref{eq:ocp-cost} and $J_s^\star$ is the value of \eqref{eq:reachable-equilibrium}, and $\xi = (x_1^\top,\ u_{2, [-n_2,-1]}^\top,\ y^\top_{2, [-n_2,-1]})^\top$ is the extended internal state equivalent to $\Sigma_c$ as in Proposition~\ref{prop:feasible-candidate}. Denote the deviation from the equilibrium as $\Delta \xi = \xi - \xi^s$. Then there exist constants $\underline c,\overline c,\overline c_\epsilon>0$ and a radius $r>0$ such that for all $\norm{\Delta \xi}_2\le r$,
\begin{equation}\label{eq:Jbounds_tracking}
\underline c\,\|\xi-\xi^s\|_2^2 
\;\le\; \widetilde J^\star(\xi)
\;\le\; \overline c\,\|\xi-\xi^s\|_2^2 \;+\; \overline c_\epsilon\,\epsilon^2 .
\end{equation}
\end{lemma}

\begin{proof}
Let $(u^s,y^s)$ be the artificial equilibrium solving \eqref{eq:reachable-equilibrium} and let $z^s$ be the associated steady state of the composite minimal realization of $\Sigma_c$. By Remark~\ref{IOSS}  and $L\ge 2n$, there exists $\alpha_\circ>0$ such that the finite-horizon observability inequality holds:
\[
\sum_{i=0}^{n_c-1}\!\|y(i)-y^s\|_2^2 \;\ge\; \alpha_o \|z(0)-z^s\|_2^2 .
\]
Since $Q\succ0$ and $z(0)-z^s$ depends linearly on $\xi-\xi^s$ (as in the Proof of Proposition \ref{prop:feasible-candidate}, there exists $\beta>0$ with
$\|z(0)-z^s\|_2 \ge \beta \|\xi-\xi^s\|_2$. Dropping nonnegative terms in the stage cost (\ref{eq:ocp-cost})
\[
\widetilde J^\star(\xi) \;\ge\; \lambda_{\min}(Q)\,\alpha_o\,\beta^2\,\|\xi-\xi^s\|_2^2 .
\]
Set $\underline c:=\lambda_{\min}(Q)\alpha_o\beta^2$ we get the lower bound in (\ref{eq:Jbounds_tracking}).
By controllability of $\Sigma_c$ and the bridge construction, the stage sum over the first $n$ steps is bounded as
\[
\sum_{i=0}^{L-1}\!\big(\|y(i)-y^s\|_Q^2+\|u(i)-u^s\|_R^2\big)\;\le\; c_{\rm st}\,\|\Delta\xi\|_2^2 .
\]
Using the feasible candidate of Proposition~\ref{prop:feasible-candidate} (Proof of $(b)$), and note
\[
\|\hat g\|_2^2 \;\le\; c_{\rm pe}\big(\|\hat u_2\|_2^2+\|\hat x_2\|_2^2\big)\;\le\; c_g\,\|\Delta\xi\|_2^2 .
\]
Similarly, writing $H_{L+n_2}(y_2^d)=H_{L+n_2}(y_2^n)+H_{L+n_2}(\delta)$ and using $\|H_{L+n_2}(\delta)\|_2\le c_\delta\,\epsilon$, the slack satisfies
\begin{align*}
\|\hat\sigma\|_2^2 \;&\le\; 2\,\|H_{L+n_2}(\delta)\|_2^2\,\|\hat g\|_2^2 + 2\,n_2\,\epsilon^2 \\ &
\;\le\; 2\,c_\delta^2 c_g\,\epsilon^2\|\Delta\xi\|_2^2 + 2\,n_2\,\epsilon^2 .
\end{align*}
Evaluating \eqref{eq:ocp-cost} at the candidate and subtracting $J_s^\star$ gives
\begin{align*}
\widetilde J^\star(\xi)\;&\le\; 
c_{\rm st}\,\|\Delta\xi\|_2^2 + \lambda_{\max}(\Lambda)\,c_g\,\|\Delta\xi\|_2^2 \\ &
+ \lambda_{\max}(\Gamma)\big(2c_\delta^2c_g\,\epsilon^2\|\Delta\xi\|_2^2 + 2n_2\,\epsilon^2\big).
\end{align*}
Hence \eqref{eq:Jbounds_tracking} holds with
$\overline c:=c_{\rm st}+\lambda_{\max}(\Lambda)c_g+2\lambda_{\max}(\Gamma)c_\delta^2c_g$ and $\overline c_\epsilon:=2\lambda_{\max}(\Gamma)n_2$. 
\end{proof}

\smallskip
Note that the upper bound in Lemma \ref{lem:optimal-cost} is precisely $V_{\max}$ in the statement of Proposition \ref{prop:feasible-candidate}. The above lemma, together with Remark \ref{IOSS}, allows us to construct a candidate Lyapunov function and state the following; omitting the proof for space since it follows similar lines to \cite[Th.~7]{Berberich2020Tracking}.
\vspace{-4pt}
\begin{proposition}[Practical exponential stability]
\label{th:practical_stability}
Suppose the hypotheses of Propositions ~\ref{prop:prediction-error}-\ref{prop:feasible-candidate}, and Lemma~\ref{lem:optimal-cost} are satisfied. Define, for a design parameter $\gamma>0$, the Lyapunov candidate $V(\xi)\;:=\;J^\star_L(\xi)+\gamma\,W(\xi-\xi_s)-J^\star_s$.
Then there exist constants $\underline\alpha,\overline\alpha>0$, $k_V\in(0,1)$, and $k_\epsilon\ge0$ and a radius $r>0$ such that for all $\xi$ with $\|\Delta\xi\|_2\le r$ and $n = \max\{n_1, n_2\}$:
\begin{enumerate}
  \item[\emph{(a)}]   $
    \underline\alpha\|\Delta\xi\|_2^2 - c_\varepsilon\varepsilon^2 \;\le\; V(\xi) \;\le\; \overline\alpha\|\Delta\xi\|_2^2 + c_\epsilon\epsilon^2 ,
  $
  \item[\emph{(b)}]  
    $V(\xi(t+n)) - V(\xi(t)) \;\le\; -\,k_V\,V(\xi(t)) + k_\epsilon\,\varepsilon^2 .
  $
\end{enumerate}
\end{proposition}

\section{Case Study: the iRonCub robot}
\label{sec:validation}

We validated the proposed DFMPC on iRonCub, a jet-powered humanoid robot built on the iCub3 platform~\cite{dafarra2024icub3}. The robot uses four jet turbines for aerial maneuvers. 
\begin{figure}[t]
  \centering
  \includegraphics[width=0.95\columnwidth,trim=2mm 1mm 2mm 1mm,clip]{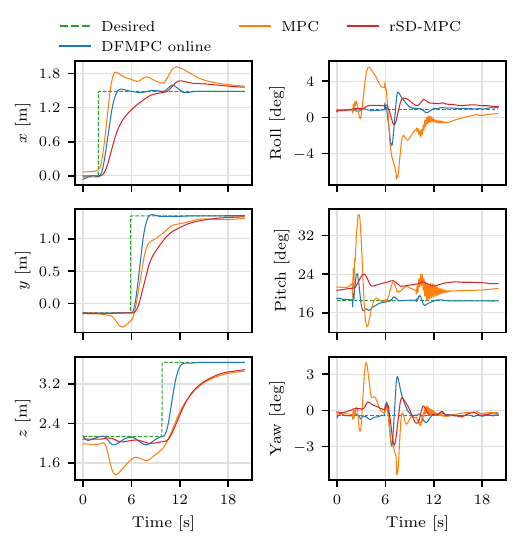}
  \caption{Trajectory tracking: CoM components $(x,y,z)$ on the left column; base orientation (roll, pitch, yaw) on the right column.}
  \label{fig:ironcub_tracking}
\end{figure}
\vspace{-6pt}
\subsection{Implementation Details}
\vspace{-4pt}
We model the robot by decomposing it into two subsystems. The known subsystem $\Sigma_{1}$ is represented by the momentum dynamics of the robot, described in \cite[Eq.~(17)]{gorbani2025unified}, while the unknown subsystem $\Sigma_{2}$ is represented by the dynamics of the jet turbines thrust.
To handle the inherent non-linearity of the turbines, we update the Hankel matrices of $\Sigma_2$ at each control step using the most recent input-output data. This online update represents a deviation from the LTI assumptions under which our theoretical guarantees in Section \ref{sec:mainResults} were derived. Consequently, the formal proofs of recursive feasibility and practical stability may not directly carry over.
However, the online update of the Hankel matrices is a well-established approach, and the rationale is that for a sufficiently fast sampling rate, the dynamics can be locally approximated by LTI models~\cite{berberich2022linear}, where the resulting linearization error is treated as a disturbance that the slack variables are designed to absorb. In the case of iRonCub, such adaptation is essential to cope with the variability of turbine behavior under different operating conditions and environmental disturbances.

The offline data trajectory ${u}_2^d$ used to build the Hankel matrices is obtained by exciting the turbines with a frequency-sweeping step inputs, which yields persistently exciting inputs of order $L+2n$ (verified by rank condition). During closed-loop operation, sufficient excitation can be maintained by adding suitable probing components to the applied input (e.g., small-amplitude perturbation or scheduled reference variations); in our simulations, the natural reference changes already provided enough excitation, so additional probing was not required. For the jet turbines, physical insight suggests a low-order response, so the value of $n_2$ was chosen as 2.\looseness=-1

To quantify the effect of the online Hankel update, we consider two variants of the proposed controller: DFMPC (online Hankel), which updates the data matrices at every control step, and DFMPC (fixed Hankel), which instead adheres to the assumptions underlying the analysis in Section III, reporting tracking RMSE and average/maximum solver runtime.\looseness=-1

\vspace{-6pt}
\subsection{Simulation Benchmarks and Ablations}
\vspace{-3pt}

{
We evaluate four controllers on iRonCub: (i) \textbf{DFMPC (online Hankel)}, (ii) \textbf{DFMPC (fixed Hankel)}, (iii) \textbf{rSD-MPC} (robust semi–data-driven MPC proposed in \cite{zieglmeier2025semi}), and (iv) the \textbf{MPC} proposed in \cite{gorbani2025unified}, which relies on a second-order linear approximation of the turbine dynamics. All controllers share the same horizons and solver settings, while the weights have been tuned for each controller. We model the measurement noise through a uniform distribution in $[-0.5,0.5]$, and set $\epsilon=0.5$. Figure~\ref{fig:ironcub_tracking} shows CoM and attitude tracking; Table~\ref{tab:benchmarks} reports horizons, runtime, and tracking errors.}

{
\textbf{Tracking:}
DFMPC (online) achieves the lowest CoM RMSE ($0.27\,\mathrm{m}$) compared with model-based MPC ($0.47\,\mathrm{m}$) and rSD-MPC ($0.28\,\mathrm{m}$). Orientation errors are also lower for all data-driven variants. The fixed-Hankel ablation attains similar rotation accuracy but worse CoM tracking.}

{
\textbf{Ablation:}
The fixed-Hankel controller, which satisfies the LTI/fixed-window assumptions, shows degraded tracking relative to the online version, confirming that online updates improve prediction quality in practice.}

{
\textbf{Runtime:}
DFMPC (online) runs in $6.4\,\mathrm{ms}$ on average ($9.7\,\mathrm{ms}$ max), remaining real-time at $100\,\mathrm{Hz}$. rSD-MPC averages $7.3\,\mathrm{ms}$ ($13.2\,\mathrm{ms}$ max), while the model-based MPC is fastest at $2.2\,\mathrm{ms}$. All controllers used OSQP as solver~\cite{stellato2020osqp}, keeping the default settings and enabling the polishing option.\looseness=-1}

\begin{table}[]
\scriptsize
\centering
\vspace{4pt}
\caption{{Hankel matrix parameters (L,N), average|max runtime, and tracking error for DFMPC (online/fixed), rSD-MPC, and MPC}}
\vspace{-5pt}
\label{tab:benchmarks}
\begin{tabular}{ccccccc}
\hline
\multirow{2}{*}{Method} & \multirow{2}{*}{L} & \multirow{2}{*}{N} & \multicolumn{2}{c}{Time [ms]} & \multicolumn{2}{c}{RMSE}         \\
                        &                    &                    & avg        & max         & CoM {[}m{]} & Rotation {[}deg{]} \\ \hline
DFMPC (online)          & 15                 & 300                & 6.38       & 9.69        & 0.27        & 0.59               \\
\rowcolor{gray!15}
DFMPC (fixed)           & 15                 & 300                & 6.59       & 11.74       & 0.36        & 0.54               \\
rSD-MPC                 & 15                 & 300                & 7.31       & 13.19       & 0.28        & 2.17               \\
\rowcolor{gray!15}
MPC                     & -                  & -                  & 2.2        & 4.3         & 0.47        & 2.46              \\ \hline
\end{tabular}
\vspace{-15pt}
\end{table}

\section{Conclusion}
\label{sec:conclusion}

This paper presented a Data-Fused Model Predictive Control (DFMPC) framework that integrates physics-based models with data-driven representations while explicitly handling measurement noise and enabling piecewise constant reference tracking. We established theoretical guarantees of recursive feasibility and practical stability under input-output constraints.
The scheme was validated on the iRonCub robot, where the momentum dynamics are well understood but the turbine dynamics remain difficult to model reliably. Simulation results demonstrated that the DFMPC improves tracking accuracy compared to a purely model-based baseline, while remaining computationally feasible for real-time implementation.
Future work will focus on formally extending the theoretical guarantees of recursive feasibility and stability to this class of adaptive data-driven controllers for nonlinear systems.\looseness=-1

\balance
\bibliographystyle{IEEEtran}
\bibliography{IEEEexample}

\end{document}